\documentclass[11pt,a4paper,reqno,oneside]{amsart}
\usepackage[utf8]{inputenc}
\usepackage[T1]{fontenc}
\usepackage[a4paper]{geometry}
\usepackage{lmodern,textcomp,amsmath,amssymb,mathtools,braket,authblk,amsthm,csquotes,bm,graphicx,mathrsfs,float,xcolor,comment,mathtools,amssymb,amsthm,braket,pgfplots,cite}
\usepackage[english]{babel}
\usepackage[foot]{amsaddr}
\usepackage[utf8]{inputenc}
\usepackage[T1]{fontenc}
\usepackage[english]{babel}
\usepackage{todonotes}
\usepackage[shortlabels]{enumitem}
\usepackage[normalem]{ulem}
\usepackage[foot]{amsaddr}

\usepackage[unicode=true,bookmarks=true,bookmarksopen=false,breaklinks=false,pdfborder={0 0 0},colorlinks=true]{hyperref}
\usepackage{xcolor}
\definecolor{cblue}{rgb}{0.16, 0.32, 0.75}
\definecolor{cred}{rgb}{0.7, 0.11, 0.11}
\hypersetup{%
	,linkcolor=cred
	,citecolor=cblue
	,urlcolor=black
}

\makeatletter
\def\@setemails{%
 \mbox{{\itshape ${}^*$Corresponding author}:\space}{\ttfamily\emails}.%
}
\makeatother

\newcommand{\hilb}{\mathcal{H}}
\newcommand{\e}{\mathrm{e}}
\newcommand{\ii}{\mathrm{i}}
\newcommand{\tr}{\operatorname{tr}}

\newtheorem{theorem}{Theorem}[section]

\newtheorem{lemma}[theorem]{Lemma}
\newtheorem{proposition}[theorem]{Proposition}
\theoremstyle{definition}
\newtheorem{definition}[theorem]{Definition}

\theoremstyle{remark}
\newtheorem{remark}[theorem]{Remark}
\newtheorem{example}[theorem]{Example}
\AtBeginEnvironment{example}{%
	\pushQED{\qed}%
}
\AtEndEnvironment{example}{\popQED\endexample}
\AtBeginEnvironment{remark}{%
	\pushQED{\qed}%
}
\AtEndEnvironment{remark}{\popQED\endexample}

\usepackage{orcidlink}

\DeclareMathOperator{\Dom}{Dom}
\DeclareMathOperator{\Ran}{Ran}

\title[On the Liouville--von Neumann Equation for unbounded Hamiltonians]{
On the Liouville--von Neumann Equation\\for unbounded Hamiltonians
}
\author{Davide Lonigro\textsuperscript{\,1,*}\hspace{2pt}\orcidlink{0000-0002-0792-8122}\hspace{1pt}}
 \email{\footnotesize \href{mailto:davide.lonigro@fau.de}{\texttt{davide.lonigro@fau.de}}}

\author{Alexander Hahn\textsuperscript{\,2}\hspace{2pt}\orcidlink{0000-0002-4152-9854}\hspace{1pt}}

\author{Daniel Burgarth\textsuperscript{\,1,2}\hspace{2pt}\orcidlink{0000-0003-4063-1264}\hspace{1pt}}

 \address{\footnotesize \textsuperscript{1}Department Physik, Friedrich-Alexander-Universität Erlangen-Nürnberg, Staudtstraße 7, 91058 Erlangen, Germany}
\address{\footnotesize \textsuperscript{2}School of Mathematical and Physical Sciences, Macquarie University, Balaclava Rd, Macquarie Park NSW 2109, Australia}

\pgfplotsset{compat=1.18}

\begin{document}

\maketitle
\thispagestyle{empty}

 \vspace{-0.5cm}
	\begin{abstract}
        The evolution of mixed states of a closed quantum system is described by a group of evolution superoperators whose infinitesimal generator (the quantum Liouville superoperator, or Liouvillian) determines the mixed-state counterpart of the Schr\"odinger equation: the Liouville--von Neumann equation. When the state space of the system is infinite-dimensional, the Liouville superoperator is unbounded whenever the corresponding Hamiltonian is. In this paper, we provide a rigorous, pedagogically-oriented, and self-contained introduction to the quantum Liouville formalism in the presence of unbounded operators. We present and discuss a characterization of the domain of the Liouville superoperator originally due to M.~Courbage; starting from that, we develop some simpler characterizations of the domain of the Liouvillian and its square. We also provide, with explicit proofs, some domains of essential self-adjointness (cores) of the Liouvillian.
    \end{abstract}


\section{Introduction}

The normal states of a quantum system are known to be described by density operators on a Hilbert space $\hilb$. Those are positive and trace-class elements of $\mathcal{B}(\hilb)$ (the space of all bounded operators defined on $\hilb$) satisfying the normalization condition $\tr\rho=1$. In general, $\hilb$ might have infinite dimension. With an abuse of notation, we shall hereafter identify the physical states of the systems with the density operator associated with them. States with $\Ran\rho=1$, and thus characterized by $\rho=
\braket{\varphi,\cdot}\varphi$ for some (unique up to a phase) $\varphi\in\hilb$, are referred to as \textit{pure} states, all other states being \textit{mixed} states.

Assuming that the system is closed, the free evolution of the system is implemented, in the Schr\"odinger picture, as follows: if $\rho_0$ is the state of the system at time $t=0$, after a time $t\in\mathbb{R}$ the state of the system is given by
\begin{equation}
    \rho(t)=U(t)\rho_0U(t)^*,
\end{equation}
where $\left(U(t)\right)_{t\in\mathbb{R}}$ is a strongly continuous one-parameter group of unitary operators (unitary propagator), that is: $U(t)U(s)=U(t+s)$ for all $t,s\in\mathbb{R}$, $U(t)^*=U(t)^{-1}$ for every $t\in\mathbb{R}$, and, given any $\varphi\in\hilb$, the function $t\mapsto U(t)\varphi$ is continuous. In particular, an initial pure state associated with a state vector $\psi_0$ evolves into the pure state associated with the state vector $\psi(t)=U(t)\psi_0$.

While the unitary propagator $U(t)$ completely characterizes the evolution of all pure or mixed states of a quantum system, more often than that one needs to reconstruct the propagator from a differential equation. In this regard, Stone's theorem establishes the following: there exists a unique, and possibly \textit{unbounded}, self-adjoint operator $H$ on $\hilb$ playing the role of the infinitesimal generator of $U(t)$. That is, the first-order differential problem ($\hbar=1$)
\begin{equation}\label{eq:schrodinger}
    \ii\frac{\mathrm{d}}{\mathrm{d}t}\varphi(t)=H\varphi(t),\qquad \varphi(0)=\varphi_0\in\Dom H,
\end{equation}
admits $\varphi(t):=U(t)\varphi_0$ as its unique solution. This is the Schr\"odinger equation, which, by construction, characterizes the evolution of all pure states of the system. Here, $\Dom H$ is the \textit{domain} of $H$. This is a crucial point: whenever $H$ is unbounded, it will only be defined on a dense but proper subspace $\Dom H$ of $\hilb$, which precisely coincides with the space of all admissible initial conditions for the Schr\"odinger equation. As such, the domain $\Dom H$ is integral to the definition of $H$, and to the physics of the problem, as much as its expression. Since $H$ is also the operator associated with the free energy of the system, one can, in principle, infer information about the evolution group $U(t)$ in terms of experimentally measurable quantities, like energy moments.

It is therefore natural to ask whether a similar differential equation can be conceived for generic states $\rho(t)$ of the system. Via a straightforward (formal) calculation, one indeed readily identifies a natural candidate of such an equation in the \textit{Liouville--von Neumann equation}:
\begin{equation}\label{eq:liouville}
    \ii\frac{\mathrm{d}}{\mathrm{d}t}\rho(t)=[H,\rho(t)],\qquad \rho(0)=\rho_0,
\end{equation}
with $[\cdot,\cdot]$ denoting the commutator, its solution being $\rho(t)=U(t)\rho_0 U(t)^*$. The quantity in the right-hand side of Eq.~\eqref{eq:liouville} can be interpreted as the action on an operator $[H,\cdot]$ on the time-dependent density operator---a \textit{superoperator}, as they are commonly referred to in the physics literature (and we will stick to this nomenclature) to distinguish them from ``ordinary'' operators on $\hilb$. We refer to Ref.~\cite{gyamfi2020fundamentals} for an extensive review of the Liouville formulation of quantum mechanics in finite-dimensional Hilbert spaces.

However, in the infinite-dimensional setting, novel issues arise:
\begin{itemize}
    \item It is not clear \textit{a priori} in which sense the derivative in the left-hand side of Eq.~\eqref{eq:liouville} should be interpreted, since one may consider different norms on $\mathcal{B}(\hilb)$ (uniform norm, trace norm, Hilbert--Schmidt norm, ...) and such norms are \textit{not} equivalent if $\dim\hilb=\infty$: the choice of norm directly affects the very existence of solutions.
    \item Furthermore, if $H$ is unbounded, it is not obvious what conditions (if any) can be imposed on the initial condition $\rho_0$ in such a way that the right-hand side of Eq.~\eqref{eq:liouville} is well-defined at all times.
\end{itemize}
Addressing these mathematical questions is of fundamental importance whenever describing infinite-dimensional quantum systems, as mixed states naturally arise in many applications---quantum information and quantum statistical mechanics, to name a few. It might thus come as a surprise that such questions are not so commonly addressed in the literature. Textbooks on the mathematical methods of quantum mechanics, while extensively covering the role of infinite-dimensional spaces and unbounded operators in quantum mechanics, mostly focus on the evolution of pure states. Ultimately, the Liouville--von Neumann equation seems to be rarely considered outside the finite-dimensional realm. Most surprisingly, the very characterization of the domain of the Liouville superoperator in terms of the domain of the corresponding Hamiltonian does not seem to be of common knowledge. This is particularly relevant for infinite-rank density operators, like the ones encountered in statistical or thermal physics.

Here, we try to address all these questions in a self-contained way, gathering some rigorous results---some of them being found in the literature, others being only implicitly hinted at---about the Liouville formalism of quantum mechanics for systems living in an infinite-dimensional Hilbert space, with their energy being described by possibly unbounded operators.

\subsection{Bibliographic remarks}

Historically, an early discussion about the Liouville formalism of quantum mechanics for systems with unbounded energy was provided by Moyal in Ref.~\cite{moyal69mean}\@. Here, the Liouvillian is first constructed as an unbounded superoperator on the Hilbert space of Hilbert--Schmidt operators on $\hilb$ (the Liouville space), corresponding to the infinitesimal generator of the evolution group $\rho\mapsto U(t)\rho U(t)^*$, and the action of this generator on its domain is (incorrectly, as discussed in later papers) identified by $[H,\rho]$.
This construction was then reprised by Spohn~\cite{spohn1975spectral,Spohn1976}, who analyzed the spectrum of the Liouville superoperator; in particular, in Ref.~\cite{Spohn1976} a core of the Liouvillian is provided, but no expression of the self-adjointness domain is given. In the same years, Prugove{\v{c}}ki and collaborators~\cite{prugovevcki1972scattering,prugovevcki1973multichannel,prugovevcki1975semi,prugovecki1982quantum} 
studied the problem in the more general context of superoperators (there named ``transformers'') on minimal norm ideals of $\mathcal{B}(\hilb)$. In particular, characterizations for the self-adjointness domain are claimed in Refs.~\cite{prugovevcki1973multichannel} and~\cite{prugovecki1982quantum}, but no proof is given. 

In Ref.~\cite{courbage1982mathematical}, a full characterization of the domain and action of the Liouville superoperator---exclusively written in terms of the properties of the corresponding Hamiltonian---was provided. As will be discussed later on, the proof partially relies on a lemma first proven in the earlier paper~\cite{courbage1971normal} from the same author, and on a previous result proven in~\cite{lanz1971existence}\@. More recently, Hirokawa~\cite{hirokawa1993rigorous} provided a construction of the Liouville superoperator for nonrelativistic quantum field theories, and a characterization of its core, on a different Hilbert space of operators endowed with a temperature-dependent inner product. 

The spectrum of the Liouville superoperator was studied in greater detail in a series of papers by Antoniou et al.~\cite{antoniou1998general,antoniou1999spectrum,antoniou2004quantum}, where, in particular, the spectral decomposition of the Liouvillian is shown to be more complex than originally claimed in Refs.~\cite{spohn1975spectral,Spohn1976}\@. As a physically relevant example, it is shown that Hamiltonians with purely singular continuous spectrum can yield Liouvillians with purely absolutely continuous spectrum: in other words, the scattering theories for the Schr\"odinger and Liouville--von Neumann equations are inequivalent. A further interesting novelty of the quantum Liouville formalism is the possibility of consistently defining time and entropy superoperators in the Liouville space, which was explored in Refs.~\cite{ordonez2001explicit,courbage1982mathematical}\@.

Aside from the unitary scenario, a comprehensive discussion on dynamical semigroups with unbounded generators was provided in Ref.~\cite{siemon2017unbounded}, where a standard form of such generators was provided and discussed.

\subsection{Scope of the paper and outline}
In this paper, we would like to provide a simple, rigorous, and self-contained introduction to the Liouville--von Neumann equation generated by unbounded Hamiltonians on infinite-dimensional Hilbert spaces. This paper should not be interpreted as a comprehensive review of the subject, neither at the mathematical nor physical level. Rather, it should be regarded as a self-contained guide to the mathematically-oriented working physicist who would like (and possibly need\footnote{Admittedly, this was the very situation in which the authors of the present manuscript found themselves (cf.~\cite{hahn2024efficiency}), and which eventually inspired them to write the present note.}) to face the aforementioned mathematical subtleties rather than sweeping them under the carpet---and, at the same time, would like to embark in this journey, at least in a first moment, with the exact degree of mathematical generality that is strictly needed to describe the Liouville--von Neumann equation rigorously. By this, we mean that all theorems will be directly presented (and, wherever possible, proven) in the Hilbert space setting, instead of deriving them---as mostly done in the aforementioned references---as particular cases deducted from the general scenario of parametric semigroups of linear transformations on a Banach space. Apart from a concise and self-contained introduction to the Liouville--von Neumann equation in infinite dimensions, this paper also reports new results on the domains of the Liouville superoperator and its powers. 

The paper is organized as follows. We start in Section~\ref{sec:setup} by recalling some basic notions and properties of (possibly) unbounded linear operators on a complex separable Hilbert space $\hilb$, and introducing the space of Hilbert--Schmidt operators on it---the Liouville space $\mathcal{L}(\hilb)$. In Section~\ref{sec:liouville} we introduce the quantum Liouville superoperator $\mathbf{H}$ associated with a quantum Hamiltonian $H$, revisit a characterization of the domain and action of the Liouvillian first due to Courbage (Theorem~\ref{thm:liouville}), and provide an equivalent, but conceptually simpler, characterization of the domain in Theorem~\ref{thm:equiv_domain}; furthermore, we generalize both results to the square of the Liouvillian (Propositions~\ref{prop:liouville^2} and~\ref{prop:equiv_domain^2}). Finally, in Section~\ref{sec:cores}, we analyze some domains of essential self-adjointness of the Liouvillian that are found in the literature (Theorem~\ref{thm:core1} and Proposition~\ref{prop:core4}). These results will enable us, in particular, to clarify what conditions one should impose on the initial state of the Liouville--von Neumann equation so that the existence and uniqueness of the solution are ensured. Final considerations are gathered in Section~\ref{sec:conclusion}\@.

\section{Mathematical set-up}\label{sec:setup}

\subsection{Preliminaries}\label{sec:prelimin} To fix the notation, we will start by providing a quick, informal overview of the mathematical formalism of quantum mechanics in infinite-dimensional Hilbert spaces. We refer to any of the excellent monographs on the subject, see e.g.\ Ref.~\cite{reed1981functional} (and subsequent volumes), or Refs.~\cite{takhtadzhian2008quantum,schmudgen2012unbounded,teschl2014mathematical,moretti2017spectral}, for details.

Throughout this paper, we will consider a complex, separable, infinite-dimensional Hilbert space $\hilb$, endowed with a scalar product $\braket{\cdot,\cdot}:\hilb\times\hilb\rightarrow\mathbb{C}$ (linear to the right) and associated norm given by $\|\varphi\|:=\sqrt{\braket{\varphi,\varphi}}$. Here, ``separable'' means that there exists one (and thus infinitely many) complete orthonormal \emph{countable} set (orthonormal basis for short) $(e_n)_{n\in\mathbb{N}}\subset\hilb$. That is, a countable set of orthonormal vectors such that $\overline{\operatorname{Span}(e_n)_{n\in\mathbb{N}}}=\hilb$, where
\begin{equation}
    \operatorname{Span}(e_n)_{n\in\mathbb{N}}=\left\{\sum_{n=1}^N\alpha_ne_n:\;N\in\mathbb{N},\alpha_1,\dots,\alpha_N\in\mathbb{C}\right\},
\end{equation}
and the overline denotes the set closure.

\medskip\indent\textbf{Bounded and unbounded linear operators.}
    A linear operator on $\hilb$ is a linear function $A:\Dom A\subset\hilb\rightarrow\hilb$, where $\Dom A$, the domain of $A$, is a vector subspace of $\hilb$. We say that $A$ is densely defined whenever $\Dom A$ is dense in $\hilb$, namely, $\overline{\Dom A}=\hilb$. Unless otherwise mentioned, all operators will be assumed to be densely defined in this paper. A densely defined operator $A$ is said to be
\begin{itemize}
    \item bounded, if $\sup_{0\neq\varphi\in\Dom A}\|A\varphi\|/\|\varphi\|=:\|A\|<\infty$;
    \item unbounded, otherwise.
\end{itemize}
In the first case, by the bounded linear extension (BLT) theorem~\cite[Theorem~I.7]{reed1981functional}, $A$ always admits a unique bounded extension to the whole Hilbert space, with the same norm. For this reason, bounded operators are usually directly defined with $\Dom A=\hilb$; however, this is not automatic---and this distinction will be indeed relevant for the present paper. The space $\mathcal{B}(\hilb)$ of bounded operators defined on the whole Hilbert space $\hilb$, endowed with the norm $A\mapsto\|A\|$ as defined above, is a Banach space~\cite[Theorem~III.2]{reed1981functional}\@.

\medskip\indent\textbf{Adjoint and closure.} The adjoint of a densely defined operator $A:\Dom A\subset\hilb\rightarrow\hilb$ is the unique operator defined as follows:\footnote{Notice that this definition identifies $\hilb$ with its dual $\hilb^*$, which is possible due to Riesz' Lemma~\cite[Theorem~1.8]{teschl2014mathematical}\@.}
\begin{align}
    \Dom A^*&=\left\{\psi\in\hilb:\exists\tilde{\psi}\in\hilb\:\text{such that}\,\braket{\psi,A\varphi}=\braket{\tilde{\psi},\varphi}\;\forall\varphi\in\Dom A\;\right\};\\
    A^*\psi&=\tilde{\psi},
\end{align}
thus satisfying $\braket{\psi,A\varphi}=\braket{A^*\psi,\varphi}$ for all $\psi\in\Dom A^*$ and $\varphi\in\Dom A$. 

A densely defined operator $A:\Dom A\subset\hilb\rightarrow\hilb$ is said to be closable if its adjoint $A^*$ is itself densely defined; if so, the closure of $A$ is defined by $\overline{A}:=A^{**}$, and $A$ is said to be closed if $\overline{A}=A$. In particular, if $A$ is bounded, then $\overline{A}\in\mathcal{B}(\hilb)$ and its closure coincides precisely with its unique extension given by the aforementioned BLT theorem. In particular, if $A\in\mathcal{B}(\hilb)$, then $\overline{A}=A$. This also follows from the closed graph theorem~\cite[Theorem~2.9]{teschl2014mathematical}, which states that an operator with domain $\Dom A=\hilb$ is bounded if and only if it is closed.

\medskip\indent\textbf{Symmetric and self-adjoint operators.}
A densely defined operator $A$ is said to be
\begin{itemize}
    \item symmetric (or Hermitian), if $A\subset A^*$;
    \item self-adjoint, if $A=A^*$,
\end{itemize}
where $A\subset A^*$ means that $\Dom A\subset\Dom A^*$ and $A^*\varphi=A\varphi$ for all $\varphi\in\Dom A$. One immediately shows that $A$ is symmetric if and only if $\braket{\psi,A\varphi}=\braket{A\psi,\varphi}$ for all $\psi,\varphi\in\Dom A$; however, this is not a sufficient condition for self-adjointness, except in the case $A\in\mathcal{B}(\hilb)$ where symmetry and self-adjointness coincide. This distinction is of physical relevance since, in quantum mechanics, physical observables are associated with self-adjoint operators. By construction, all symmetric operators are closable and all self-adjoint operators are closed. It is also worth to remark that, by the Hellinger--Toeplitz theorem~\cite[Theorem~2.10]{teschl2014mathematical}, a symmetric (and, in particular, a self-adjoint) operator $A$ with $\Dom A=\hilb$ is necessarily bounded.

Furthermore, $A$ is said to be essentially self-adjoint if its closure $\overline{A}=A^{**}$ is self-adjoint. This happens precisely when $A$ is symmetric and admits a \textit{unique} self-adjoint extension---namely, its closure. Given a self-adjoint operator $A$, a subspace $\mathcal{D}\subset\Dom A$ is said to be a \textit{core} for $A$ if the restriction of $A$ to $\mathcal{D}$ is essentially self-adjoint, and thus, $\overline{A\restriction\mathcal{D}}=A$. A useful criterion for essential self-adjointness, due to Nelson, will be recalled in Section~\ref{sec:cores}, see Proposition~\ref{prop:nelson}\@.

\medskip\indent\textbf{Functional calculus.}
If $A$ is a self-adjoint operator, the spectral theorem for unbounded self-adjoint operator ensures that there exists a unique projection-valued measure $E_A:\mathfrak{B}(\mathbb{R})\rightarrow\mathcal{B}(\hilb)$, where $\mathfrak{B}(\mathbb{R})$ is the Borel $\sigma$-algebra on $\mathbb{R}$ (see Ref.~\cite[Chapter~A.1]{teschl2014mathematical}), such that
\begin{equation}\label{eq:operator_spectral}
    A=\int_{\mathbb{R}}\lambda\,\mathrm{d}E_A(\lambda)
\end{equation}
and
\begin{equation}\label{eq:domain_spectral}
    \Dom A=\left\{\psi\in\hilb:\int\vert\lambda\vert^2\,\braket{\psi,\mathrm{d}E_A(\lambda)\psi}<\infty\right\},
\end{equation}
with $\|A\psi\|^2$ being equal to the integral in the right-hand side of Eq.~\eqref{eq:domain_spectral}\@. By using this formalism, one can consistently define functions $f(A)$ of the operator $A$ by replacing $\lambda$ with $f(\lambda)$ in Eqs.~\eqref{eq:operator_spectral} and~\eqref{eq:domain_spectral}\@.

\medskip\indent\textbf{Unitary propagators.}
Finally, a (homogeneous) unitary propagator is a strongly continuous one-parameter family of unitary operators $\{U(t)\}_{t\in\mathbb{R}}$, where strongly continuous means that, for every $\psi\in\hilb$, the function $\mathbb{R}\ni t\mapsto U(t)\psi\in\hilb$ is continuous. Then, by the Stone theorem~\cite[Theorem 5.3]{teschl2014mathematical}, the operator $A:\Dom H\subset\hilb\rightarrow\hilb$ defined by
\begin{align}
    \Dom A&=\left\{\varphi\in\hilb:\exists\lim_{t\to0}\frac{U(t)\varphi-\varphi}{t}\in\hilb\right\};\\
    A\varphi&=\ii\lim_{t\to0}\frac{U(t)\varphi-\varphi}{t},
\end{align}
is self-adjoint; vice versa, every self-adjoint operator $A$ uniquely defines a unitary propagator via $U(t)=\e^{-\ii tA}$, where the exponentiation is defined by means of functional calculus:
\begin{equation}
    \e^{-\ii tA}=\int_{\mathbb{R}}\e^{-\ii\lambda t}\,\mathrm{d}E_A(\lambda).
\end{equation}
We say that $A$ is the infinitesimal generator of $U(t)$. By construction, the unitary propagator $U(t)$ provides the unique solution to the Schr\"odinger equation~\eqref{eq:schrodinger} as discussed in the introduction.

\subsection{Hilbert--Schmidt operators and the Liouville space} We will now introduce the main setting of this paper, namely, the space of Hilbert--Schmidt operators on $\hilb$. For completeness, all following well-known results are presented with proofs.

\begin{definition}[Hilbert--Schmidt operator]\label{def:hilbert_schmidt}
   $A\in\mathcal{B}(\hilb)$ is said to be a \textit{Hilbert--Schmidt operator} if there exists an orthonormal basis $(e_n)_{n\in\mathbb{N}}\subset\hilb$ such that
    \begin{equation}\label{eq:hs_ineq}
        \sum_{n\in\mathbb{N}}\|Ae_n\|^2<\infty.
    \end{equation}
\end{definition}
In fact, the sum in Eq.~\eqref{eq:hs_ineq} is invariant under change of basis:
\begin{proposition}\label{prop:bases}
    The following properties hold:
    \begin{itemize}
    \item[(i)] $A$ is Hilbert--Schmidt if and only if $A^*$ is Hilbert--Schmidt;
        \item[(ii)] $A$ is Hilbert--Schmidt if and only if Eq.~\eqref{eq:hs_ineq} holds for \textit{any} orthonormal basis;      
        \item[(iii)] If $A$ is Hilbert--Schmidt, then
        \begin{equation}
            \|A\|\leq\Bigl(\sum_{n\in\mathbb{N}}\|Ae_n\|^2\Bigr)^{1/2}.
        \end{equation}
    \end{itemize}
\end{proposition}
\begin{proof} Let $(e'_m)_{m\in\mathbb{N}}\subset\hilb$ another orthonormal basis. Then, for every $n\in\mathbb{N}$,
\begin{equation}
    \|Ae_n\|^2=\sum_{m\in\mathbb{N}}|\braket{e'_m,Ae_n}|^2,
\end{equation}
whence
\begin{align}\label{eq:hs1}
\sum_{n\in\mathbb{N}}\|Ae_n\|^2&=\sum_{n,m\in\mathbb{N}}|\braket{e'_m,Ae_n}|^2\nonumber\\
&=\sum_{n,m\in\mathbb{N}}|\braket{A^*e'_m,e_n}|^2\nonumber\\&=\sum_{m\in\mathbb{N}}\|A^*e'_m\|^2,
\end{align}
and in particular, choosing $e'_n=e_n$,
\begin{equation}
    \sum_{n\in\mathbb{N}}\|Ae_n\|^2=\sum_{n\in\mathbb{N}}\|A^*e_n\|^2,
\end{equation}
which is (i). To prove (ii), we apply Eq.~\eqref{eq:hs1} to $A^*$; recalling that $A^{**}=A$ for $A\in\mathcal{B}(\hilb)$, we get
\begin{equation}\label{eq:hs2}
    \sum_{n\in\mathbb{N}}\|A^*e_n\|^2=\sum_{m\in\mathbb{N}}\|Ae'_m\|^2.
\end{equation}
Eqs.~\eqref{eq:hs1} and~\eqref{eq:hs2} imply (ii). Finally, let $\psi\in\hilb$: then
\begin{align}
    \|A\psi\|^2&=\sum_{n\in\mathbb{N}}\left|\braket{e_n,A\psi}\right|^2=\sum_{n\in\mathbb{N}}\left|\braket{A^*e_n,\psi}\right|^2\nonumber\\
    &\leq\|\psi\|^2\sum_{n\in\mathbb{N}}\|A^*e_n\|^2=\|\psi\|^2\sum_{n\in\mathbb{N}}\|Ae_n\|^2,
\end{align}
where the Cauchy--Schwarz inequality was applied; whence (iii) follows.
\end{proof}

One readily sees that the sum of Hilbert--Schmidt operators is a Hilbert--Schmidt operator, whence the set of Hilbert--Schmidt operators on $\hilb$ is a vector subspace of $\mathcal{B}(\hilb)$. Furthermore, the (basis-invariant) quantity
\begin{equation}
\|A\|_{\rm HS}:=\Bigl(\sum_{n\in\mathbb{N}}\|Ae_n\|^2\Bigr)^{1/2}
\end{equation}
defines a complete norm on this vector space, with $\|A\|\leq\|A\|_{\rm HS}$. One can even define an inner product between Hilbert--Schmidt operators:
\begin{equation}\label{eq:hs_product}
    \braket{A,B}_{\rm HS}:=\sum_{n\in\mathbb{N}}\braket{Ae_n,Be_n},
\end{equation}
which is well-defined since it is basis-invariant (this being proven with an analogous argument as in Proposition~\ref{prop:bases}), and satisfies
\begin{equation}
    \|A\|^2_{\rm HS}=\braket{A,A}_{\rm HS},\qquad\left|\braket{A,B}_{\rm HS}\right|\leq\|A\|_{\rm HS}\|B\|_{\rm HS}.
\end{equation}
Therefore, Hilbert--Schmidt operators not only form a Banach subspace of $\mathcal{B}(\hilb)$: they form a Hilbert space---differently from $\mathcal{B}(\hilb)$ itself. 
This leads us to the following definition:
\begin{definition}
    The \textit{Liouville space} $\mathcal{L}(\hilb)$ is the Hilbert space of all Hilbert--Schmidt operators on $\hilb$ endowed with the scalar product $\braket{\cdot,\cdot}_{\rm HS}$ as defined in Eq.~\eqref{eq:hs_product}\@.
\end{definition}
Recall that $A\in\mathcal{B}(\hilb)$ is said to be \textit{finite-rank} if $\dim\Ran A<\infty$. Clearly, all finite-rank operators are in $\mathcal{L}(\hilb)$. In fact, the following approximation property holds:
\begin{proposition}\label{prop:finite-rank}
    The space of all finite-rank operators is dense in $\mathcal{L}(\hilb)$, that is: for all $A\in\mathcal{L}(\hilb)$ there exists a sequence $(A_n)_{n\in\mathbb{N}}$ of finite-rank operators such that $\|A-A_n\|_{\rm HS}\to0$.
\end{proposition}
\begin{proof}
    Given any orthonormal basis $(e_k)_{k\in\mathbb{N}}$, let $P_n$ be the orthogonal projection onto the linear span of $\{e_1,\dots,e_n\}$, and set $A_n:=P_nA$. This is a finite-rank operator which satisfies
    \begin{equation}
        \braket{e_m,A_n\psi}=\braket{P_ne_m,A\psi}=\begin{cases}
            \braket{e_m,A\psi},&m\leq n\\
            0,                 &m>n
        \end{cases}
    \end{equation}
for all $\psi\in\hilb$. Then
    \begin{align}
    \left\|A-A_n\right\|^2_{\rm HS}&=\sum_{h\in\mathbb{N}}\|Ae_h-A_ne_h\|^2\nonumber\\
    &=\sum_{h,m\in\mathbb{N}}|\braket{e_m,Ae_h}-\braket{e_m,A_ne_h}|^2\nonumber\\
    &=\sum_{h\in\mathbb{N}}\sum_{m=n+1}^\infty|\braket{e_m,Ae_h}|^2\nonumber\\
    &=\sum_{m=n+1}^\infty\sum_{h\in\mathbb{N}}|\braket{A^*e_m,e_h}|^2\nonumber\\
    &=\sum_{m=n+1}^\infty\|A^*e_m\|^2,
    \end{align}
    and the latter quantity converges to zero as $n\to\infty$ since $\sum_{m\in\mathbb{N}}\|A^*e_m\|^2=\|A^*\|_{\rm HS}=\|A\|_{\rm HS}$ is finite.
\end{proof}
In particular, since $\|A-A_n\|\leq\|A-A_n\|_{\rm HS}$, this also implies that all Hilbert--Schmidt operators are compact.

\begin{remark}[Vectorization]\label{rem:vectorization}
    Proposition~\ref{prop:finite-rank} is also at the root of the following observation: the space $\mathcal{L}(\hilb)$ of Hilbert--Schmidt operators on $\hilb$ is isomorphic to the Hilbert space tensor product\footnote{
    Here, the nomenclature ``Hilbert space tensor product'' stresses the fact that $\hilb\otimes\hilb^*$ must be interpreted as the \textit{completion} of the algebraic tensor product of the two spaces--that is, the completion of the space of finite tensor products of elements of the two spaces.
    } $\hilb\otimes\hilb^*$, where $\hilb^*$ is the dual of $\hilb$. The isomorphism is defined as follows: rank-one operators $\braket{\psi,\cdot}\varphi$ are mapped into $\varphi\otimes\braket{\psi,\cdot}$; this transformation is then generalized to finite-rank operators by linearity, and extended to all Hilbert--Schmidt operators by continuity. This procedure is akin to the \textit{vectorization} of (density) matrices in finite-dimensional vector spaces, see e.g.\ Ref.~\cite[Chapter~1.1]{Watrous2018} and Ref.~\cite{Wood2015}\@.
    \end{remark}

To conclude this section, let us verify that $\mathcal{L}(\hilb)$ does include all trace-class operators---and, in particular, all density operators $\rho$ encountered in quantum mechanics. Recall the following definition:
\begin{definition}
    $A\in\mathcal{B}(\hilb)$ is said to be a \textit{trace-class operator} if there exists an orthonormal basis $(e_n)_{n\in\mathbb{N}}$ such that
    \begin{equation}
        \tr(|A|):=\sum_{n\in\mathbb{N}}\braket{e_n,|A|e_n}<\infty,
    \end{equation}
    where $|A|:=(A^*A)^{1/2}$. In particular, a \textit{density operator} $\rho\in\mathcal{B}(\hilb)$ is a trace-class operator satisfying $\rho^*=\rho\geq0$ (thus $|\rho|=\rho$) and $\tr\rho=1$.\footnote{Throughout the paper, we will reserve the symbol $\rho$ to density operators only, using capital Latin letters for generic elements of $\mathcal{L}(\hilb)$.}
\end{definition}
Note that the square root in the definition above is unambiguously defined since $A^*A$ is a nonnegative operator. Trace-class operators enjoy similar properties as Hilbert--Schmidt ones, with analogous proofs:
\begin{itemize}
    \item $\tr|A|$ is independent of the particular choice of basis, and defines a norm (trace norm);
    \item finite-rank operators are dense in the space of trace-class operators with respect to the trace norm, implying that all trace-class operators are compact.
\end{itemize}
However, the space of trace-class operators is \textit{not} a Hilbert space. 

\begin{proposition}\label{prop:traceclass_are_hs}
    All trace-class operators are Hilbert--Schmidt operators.
\end{proposition}
\begin{proof}
Let $A\in\mathcal{B}(\hilb)$ be trace-class. Since $|A|$ is self-adjoint and compact, we can choose $(e_n)_{n\in\mathbb{N}}$ as an orthonormal eigenbasis of $|A|$, in which case
\begin{align}
    \sum_{n\in\mathbb{N}}\|Ae_n\|&=\sum_{n\in\mathbb{N}}\||A|e_n\|\nonumber\\
    &=\sum_{n\in\mathbb{N}}\braket{e_n,|A|e_n}<\infty.
\end{align}
This implies $\|Ae_n\|\to0$ as $n\to\infty$, and thus there exists $N\in\mathbb{N}$ such that $\|Ae_n\|<1$, and thus $\|Ae_n\|^2\leq\|Ae_n\|$, for all $n\geq N$, which implies that $\sum_{n\in\mathbb{N}}\|Ae_n\|^2$ is also finite.
\end{proof}

We finally remark that, by construction, $A\in\mathcal{B}(\hilb)$ is Hilbert--Schmidt if and only if the nonnegative operator $A^*A$ is trace-class, in which case
\begin{equation}
    \|A\|^2_{\rm HS}=\tr(A^*A);
\end{equation}
more generally, given two Hilbert--Schmidt operators $A,B$, one has
\begin{equation}\label{eq:hs_as_trace}
    \braket{A,B}_{\rm HS}=\tr(A^*B).
\end{equation}

\section{The quantum Liouville superoperator}\label{sec:liouville}

We have shown that the Liouville space $\mathcal{L}(\hilb)$ of Hilbert--Schmidt operators, endowed with the scalar product $\braket{\cdot,\cdot}_{\rm HS}$ defined in Eq.~\eqref{eq:hs_product}, is itself a Hilbert space. This crucially implies that all definitions and propositions recalled in Section~\ref{sec:prelimin} are immediately valid for linear operators on $\mathcal{L}(\hilb)$, as long as all scalar products and norms are replaced with their Hilbert--Schmidt counterparts. 

To avoid confusion, following a common convention in the physical literature, linear operators on $\mathcal{L}(\hilb)$ will be denoted as \textit{superoperators}, to distinguish them from operators on $\hilb$. Besides, to further stress this distinction, we will denote superoperators with boldface letters,
\begin{equation}
    \mathbf{A}:\Dom\mathbf{A}\subset\mathcal{L}(\hilb)\rightarrow\mathcal{L}(\hilb).
\end{equation}
In particular, the identity superoperator on $\mathcal{L}(\hilb)$ shall be denoted by $\mathbf{I}$.

\subsection{Quantum evolution in the Liouville space} Let $H:\Dom H\subset\hilb\rightarrow\hilb$ be a self-adjoint operator on $\hilb$. As recalled in Section~\ref{sec:prelimin}, $H$ is uniquely associated with a unitary propagator $(U(t))_{t\in\mathbb{R}}$, via $U(t)=\e^{-\ii tH}$, satisfying the associated Schr\"odinger equation~\eqref{eq:schrodinger}\@. In quantum mechanics, given a normalized vector $\psi\in\hilb$, the function $t\mapsto U(t)\psi$ is interpreted as the evolution of the pure state associated with $\psi$. The evolution of mixed states, associated with density operators $\rho$, is then prescribed to be given by
\begin{equation}
    t\mapsto U(t)\rho U(t)^*.
\end{equation}
Recalling that all density operators are Hilbert--Schmidt (Proposition~\ref{prop:traceclass_are_hs}), this leads us to introduce a family of superoperators $(\mathbf{U}(t))_{t\in\mathbb{R}}$ on $\mathcal{L}(\hilb)$, with $\mathbf{U}(t)A:=U(t)AU(t)^*$.

\begin{proposition}
    Let $(U(t))_{t\in\mathbb{R}}$ a unitary propagator on $\hilb$. Then the family of superoperators $(\mathbf{U}(t))_{t\in\mathbb{R}}$, with
    \begin{equation}
        \mathbf{U}(t):\mathcal{L}(\hilb)\rightarrow\mathcal{L}(\hilb),\qquad \mathbf{U}(t)A:=U(t)AU(t)^*,
    \end{equation}
    is a unitary propagator on $\mathcal{L}(\hilb)$.
\end{proposition}
\begin{proof}
    To begin with, we must verify that, for every $t\in\mathbb{R}$, $\mathbf{U}(t)$ is a well-defined superoperator on $\mathcal{L}(\hilb)$; that is, for every $A\in\mathcal{L}(\hilb)$, we also have $U(t)AU(t)^*\in\mathcal{L}(\hilb)$. Indeed,
 \begin{align}
	\sum_{n\in\mathbb{N}}\left\|U(t)AU(t)^*e_n\right\|^2&=\sum_{n,m\in\mathbb{N}}\left|\Braket{e_m,U(t)AU(t)^*e_n}\right|^2\nonumber\\
	&=\sum_{n,m\in\mathbb{N}}\left|\Braket{U(t)^*e_m,AU(t)^*e_n}\right|^2\nonumber\\&=\sum_{n\in\mathbb{N}}\left\|AU(t)^*e_n\right\|^2\nonumber\\&=\|A\|_{\rm HS}^2<\infty,
\end{align}
where we have used the fact that, for every $t\in\mathbb{R}$, $(U(t)^*e_n)_{n\in\mathbb{N}}$ is an orthonormal basis of $\hilb$ (because of the unitarity of $U(t)^*$), and the independence of the Hilbert--Schmidt norm of the choice of basis. This also proves that $\mathbf{U}(t)$ is an isometry on $\mathcal{L}(\hilb)$, i.e., $\|\mathbf{U}(t)A\|_{\rm HS}=\|A\|_{\rm HS}$. To prove that it is actually a unitary superoperator, we must compute its adjoint $\mathbf{U}(t)^*$. The latter is the superoperator uniquely characterized by
\begin{equation}
    	\braket{\mathbf{U}(t)^*A,B}_{\rm HS}=\braket{A,\mathbf{U}(t)B}_{\rm HS}\qquad\text{for all }A,B\in\mathcal{L}(\hilb);
\end{equation}
but
\begin{align}
    \braket{A,\mathbf{U}(t)B}_{\rm HS}&=\sum_{n\in\mathbb{N}}\braket{Ae_n,U(t)BU(t)^*e_n}\nonumber\\&=\sum_{n\in\mathbb{N}}\braket{U(t)^*AU(t)U(t)^*e_n,BU(t)^*e_n}\nonumber\\&=\braket{U(t)^*AU(t),B}_{\rm HS},
\end{align}
where again we used the fact that $(U(t)^*e_n)_{n\in\mathbb{N}}$ is an orthonormal basis. Therefore, $\mathbf{U}(t)^*$ is the superoperator defined by $\mathbf{U}(t)^*A:=U(t)^*AU(t)$; whence
\begin{equation}
    \mathbf{U}(t)\mathbf{U}(t)^*=\mathbf{I}=\mathbf{U}(t)\mathbf{U}(t)^*,
\end{equation}
so that $(\mathbf{U}(t))_{t\in\mathbb{R}}$ is indeed a family of unitary superoperators. Besides,
\begin{align}
    \mathbf{U}(t+s)A&=U(t+s)AU(t+s)^*\nonumber\\
    &=U(t)U(s)A\left[U(t)U(s)\right]^*\nonumber\\
    &=U(t)U(s)AU(s)^*U(t)^*\nonumber\\
    &=\mathbf{U}(t)\mathbf{U}(s)A,
\end{align}
thus the group property $\mathbf{U}(t+s)=\mathbf{U}(t)\mathbf{U}(s)$ also holds. 

To conclude, we must verify strong continuity:\footnote{Notice that, here and in the following, the notion of strong convergence in the Liouville space is to be intended with respect to Hilbert--Schmidt operators, and not with respect to Hilbert space vectors as in the pure state setting.} for every $A\in\mathcal{L}(\hilb)$, $\|\mathbf{U}(t)A-A\|_{\rm HS}\to0$ as $t\to0$. Indeed, we have
\begin{align}
	\left[\mathbf{U}(t)-\mathbf{I}\right]A&=U(t)AU(t)^*-A\nonumber\\&=U(t)AU(t)^*-U(t)A+U(t)A-A\nonumber\\
	&=U(t)A\left[U(t)-I\right]^*+\left[U(t)-I\right]A.
\end{align}
Both terms in the equation above converge to zero as $t\to0$ because of the strong continuity of $U(t)$. For example,
\begin{align}
    \|[U(t)-I]A\|_{\rm HS}^2&=\sum_{n\in\mathbb{N}}\left\|[U(t)-I]Ae_n\right\|^2,
\end{align}
and, by dominated convergence,
\begin{equation}
    \lim_{t\to0} \|[U(t)-I]A\|_{\rm HS}^2=\sum_{n\in\mathbb{N}}\lim_{t\to0}\left\|[U(t)-I]Ae_n\right\|^2=0;
\end{equation}
same for the other term.
\end{proof}

\begin{remark}
We have shown that $\mathbf{U}(t)$ is unitary, and thus an isometry, in the Hilbert--Schmidt norm. This, \textit{per se}, does not necessarily imply that $\mathbf{U}(t)$ also transforms density operators into density operators, since those entail a normalization condition in the trace norm, $\tr\rho=1$, rather than in the Hilbert--Schmidt norm. However, as well-known in quantum mechanics, this is indeed the case: given a density operator $\rho$, we have $\tr\mathbf{U}(t)\rho=\tr\rho=1$
as an immediate consequence of the cyclicity of the trace, and it is also immediate to see that $\mathbf{U}(t)\rho$ is still nonnegative and self-adjoint.
\end{remark}

Since $\left(\mathbf{U}(t)\right)_{t\in\mathbb{R}}$ is a unitary propagator, it admits an infinitesimal generator, which we will denote by $\mathbf{H}$. This is the superoperator that we will refer to as the quantum Liouville superoperator:
\begin{definition}[Quantum Liouville superoperator]\label{def:liouville}
    Let $(U(t))_{t\in\mathbb{R}}$ a unitary propagator on $\hilb$, and $(\mathbf{U}(t))_{t\in\mathbb{R}}$ the corresponding unitary propagator on $\mathcal{L}(\hilb)$ as defined before. 
    The \textit{quantum Liouville superoperator} is the infinitesimal generator of $\mathbf{U}(t)$, that is, the self-adjoint superoperator $\mathbf{H}$ on $\mathcal{L}(\hilb)$ defined by
    \begin{align}
   \Dom \mathbf{H}&=\left\{A\in\mathcal{L}(\hilb):\exists\lim_{t\to0}\frac{\mathbf{U}(t)A-A}{t}\in\mathcal{L}(\hilb)\right\};\\
    \mathbf{H}A&=\ii\lim_{t\to0}\frac{\mathbf{U}(t)A-A}{t}.
\end{align}
\end{definition}
All limits here are understood to be taken in the Hilbert--Schmidt norm, that is, $\mathbf{H}A$ is defined by
\begin{equation}\label{eq:hs_limit}
    \lim_{t\to0}\left\|\mathbf{H}A-\ii\frac{\mathbf{U}(t)A-A}{t}\right\|_{\rm HS}=0.
\end{equation}
By construction, we have $\mathbf{U}(t)=\e^{-\ii t\mathbf{H}}$, where the exponentiation is to be interpreted in the sense of functional calculus.\footnote{In fact, it is possible to characterize the projection-valued measure $\mathbf{E}_{\mathbf{H}}$ (cf.\ Section~\ref{sec:prelimin}) associated with the Liouville superoperator $\mathbf{H}$ in terms of the projection-valued measure $E_H$ associated with the corresponding Hamiltonian $H$. This allows one to relate the functional calculus for Liouvillian superoperators in terms of the one for the corresponding Hamiltonian---and thus, to relate the spectra of the two operators. This topic is extensively covered by the existing literature, cf.\ Refs.~\cite{spohn1975spectral,prugovecki1982quantum,hirokawa1993rigorous,antoniou1998general,antoniou1999spectrum,antoniou2004quantum}.
}
The Liouville space counterpart of Eq.~\eqref{eq:schrodinger} becomes
\begin{equation}\label{eq:superschrodinger}
    \ii\frac{\mathrm{d}}{\mathrm{d}t}A(t)=\mathbf{H}A(t),\qquad A(0)=A_0\in\Dom\mathbf{H},
\end{equation}
where the derivative is again understood in the Hilbert--Schmidt sense. It admits $A(t)=\mathbf{U}(t)A_0=U(t)A_0U(t)^*$ as its unique solution.

\begin{remark}[Unitary propagator and Liouvillian in the vectorized representation]\label{rem:vectorization2}
    Recall that, as pointed out in Remark~\ref{rem:vectorization}, one can equivalently represent the Liouville space $\mathcal{L}(\hilb)$ as the Hilbert space tensor product of $\hilb$ and $\hilb^*$. From the definition of said isomorphism one sees that, in this representation, $\mathbf{U}(t)$ is given by
    \begin{equation}\label{eq:propagator_vectorized}
        \mathbf{U}(t)=U(t)\otimes U(t)^*.
    \end{equation}
    One could then proceed by defining the Liouvillian as the infinitesimal generator of the unitary propagator in Eq.~\eqref{eq:propagator_vectorized}, and equivalently formulate all results in the next sections within this representation. We shall briefly comment on this point later on in the paper. In particular, see Remark~\ref{rem:vectorization3}\@.
\end{remark}

\subsection{Characterization of the quantum Liouville superoperator}

The Liouville superoperator defined above is, by its very definition, the direct counterpart of the Hamiltonian $H$ for mixed states: all basic properties of the dynamics of pure quantum states that exclusively pertain to the relation between $U(t)$ and $H$ immediately extend to mixed states. However, to make use of such results in applications, it would be useful to actually have an explicit \textit{characterization} of the superoperator $\mathbf{H}$ in terms of the corresponding Hamiltonian $H$, both in terms of the domain as well as its expression. This is because, in quantum mechanics, physical theories are usually constructed by \textit{first} constructing a self-adjoint (or essentially self-adjoint) Hamiltonian $H$ and only then computing the corresponding evolution group $U(t)$. If we wish to make use of the information at our disposal about $H$ for the evolution of mixed states as well, we need to \textit{characterize} the domain and action of $\mathbf{H}$ in terms of the ones of $H$. We will indeed present such a characterization in the present section (Theorem~\ref{thm:liouville} and Theorem~\ref{thm:equiv_domain}). Before jumping to that, let us provide an informal derivation of the expected result. 

\begin{remark}[Informal justification of Theorem~\ref{thm:liouville}]
Comparing Eq.~\eqref{eq:superschrodinger} with the Liouville--von Neumann equation~\eqref{eq:liouville}, we might expect the action of $\mathbf{H}$ to act on suitable Hilbert--Schmidt operators $A$ as $\mathbf{H}A=[H,A]=HA-AH$; besides, since we are working in $\mathcal{L}(\hilb)$, $A$ should be chosen in such a way that $[H,A]$ is itself Hilbert--Schmidt. However, if $H$ is unbounded, we cannot even have $[H,A]\in\mathcal{B}(\hilb)$ since
\begin{itemize}
    \item the operator $HA$ has domain $\Dom HA=\left\{\varphi\in\hilb:A\varphi\in\Dom H\right\}$,
    \item the operator $AH$ has domain $\Dom AH=\Dom H$,
\end{itemize}
whence $\Dom[H,A]=\Dom HA\cap\Dom AH$ is not the whole space $\hilb$. In fact, without further requirements on $A$, $\Dom[H,A]$ might not be dense in $\hilb$ (and might even only consist of the zero vector!). This can be avoided by requiring $A$ to enjoy the property $A\Dom H\subset \Dom H$, which then guarantees $\Dom HA=\Dom H$ and thus $\Dom[H,A]=\Dom H$. With this requirement, $[H,A]$ is a densely defined operator; we can then additionally require that $A$ is chosen in such a way that $[H,A]$ is actually bounded, and that its unique extension to an operator in $\mathcal{B}(\hilb)$, i.e.\ its closure $\overline{[H,A]}$ (cf.\ Section~\ref{sec:prelimin}), is a Hilbert--Schmidt operator.
\end{remark}

Indeed, the result informally presented above stands true. 
\begin{theorem}\label{thm:liouville}
\emph{(\!\!\cite[Appendix A]{courbage1982mathematical})} Let $\mathbf{H}$ be the Liouville superoperator as per Definition~\ref{def:liouville}, and $H$ the corresponding Hamiltonian. Then
\begin{align}\label{eq:domain_courbage}
    \Dom\mathbf{H}&=\left\{A\in\mathcal{L}(\hilb):\;A\Dom H\subset\Dom H,\;\overline{[H,A]}\in\mathcal{L}(\hilb)\right\};\\
    \mathbf{H}A&=\overline{[H,A]}.
\end{align}
\end{theorem}
As announced in the introduction, to our knowledge, this result was claimed (in a slightly different form) in the introduction of Ref.~\cite{courbage1982mathematical}\@. Its proof can be divided into two parts:
\begin{enumerate}
    \item[(i)] One must prove that all vectors satisfying Eq.~\eqref{eq:domain_courbage} are in the domain of $\mathbf{H}$ as in Definition~\ref{def:liouville};
    \item[(ii)] Vice versa, one must prove that all vectors in the domain of $\mathbf{H}$ as per Definition~\ref{def:liouville} satisfy the conditions in Eq.~\eqref{eq:domain_courbage}, and that, for all such vectors, $\mathbf{H}A=\overline{[H,A]}$.
\end{enumerate}
Point (i) is proven in Appendix~A of Ref.~\cite{courbage1982mathematical} and relies on the earlier result~\cite[Lemma~1]{courbage1971normal}, while the proof of point (ii) can be found in Appendix~IIa of Ref.~\cite{lanz1971existence}\@. 

For the convenience of the reader, we will present the full proof, mostly following the original sources. Let us start with the aforementioned lemma:
\begin{lemma}\label{lemma:courbage}
    \emph{(\!\!~\cite[Lemma~1]{courbage1971normal})} Let $\mathcal{K}$ be a Hilbert space and $\{V(t)\}_{t\in\mathbb{R}}$ be a strongly continuous one-parameter unitary group with infinitesimal generator $K$.\footnote{We use this generic notation here (as well as later on in Proposition~\ref{prop:nelson}) to stress that the lemma can be applied either to the Hamiltonian $H$ on $\hilb$, or to the corresponding Liouvillian $\mathbf{H}$ on $\mathcal{L}(\hilb)$.} Then the following conditions are equivalent:
    \begin{itemize}
        \item[(1)] $\Psi\in\Dom K$;
        \item[(2)] there exists $C_{\Psi}>0$ such that
    \begin{equation}\label{eq:courbage_estimate}
        \left\|V(t)\Psi-\Psi\right\|\leq C_{\Psi}|t|
    \end{equation}
    for all $t\in\mathbb{R}$.
    \end{itemize}
\end{lemma}
\begin{proof}
We will follow the proof of Ref.~\cite[Lemma~1]{courbage1971normal}\@. By the spectral theorem (cf.\ Section~\ref{sec:prelimin}), we have
\begin{equation}
    V(t)=\int_{\mathbb{R}}\e^{-\ii\lambda t}\,\mathrm{d}E_K(\lambda),
\end{equation}
with $E_K$ being the projection-valued measure associated with $K$.\medskip

(1)$\!\implies\!$(2). As recalled in Section~\ref{sec:prelimin}, this happens if and only if
\begin{equation}
    \|K\Psi\|^2=\int_\mathbb{R}\lambda^2\,\braket{\Psi,\mathrm{d}E_K(\lambda)\Psi}<\infty.
\end{equation}
But then, noticing that $|e^{-\ii\lambda t}-1|^2=4\sin^2(\lambda t/2)\leq\lambda^2t^2$,
\begin{align}
    \left\|(V(t)-I)\Psi\right\|^2&=\int_{\mathbb{R}}|\e^{-\ii\lambda t}-1|^2\,\braket{\Psi,\mathrm{d}E_K(\lambda)\Psi}\nonumber\\
    &\leq t^2\int_{\mathbb{R}}\lambda^2\,\braket{\Psi,\mathrm{d}E_K(\lambda)\Psi},
\end{align}
whence, if $\Psi\in\Dom K$, the estimate~\eqref{eq:courbage_estimate} holds with $C_\Psi=\|K\Psi\|$.\medskip

(2)$\!\implies\!$(1). We will prove that, if (1) is false, then (2) is also false. Let $\Psi\notin\Dom K$; this means that, choosing an arbitrary positive constant $M>0$, there exists $\Lambda_M>0$ such that
\begin{equation}
    \int_{[-\Lambda_M,\Lambda_M]}\lambda^2\,\braket{\Psi,\mathrm{d}E_K(\lambda)\Psi}>2M^2\|\Psi\|^2.
\end{equation}
Besides,
\begin{align}
    \left\|(V(t)-I)\Psi\right\|^2&=\int_\mathbb{R}|\e^{-\ii\lambda t}-1|^2\,\braket{\Psi,E_K(\lambda)\Psi}\nonumber\\
    &\geq\int_{[-\Lambda_M,\Lambda_M]}|\e^{-\ii\lambda t}-1|^2\,\braket{\Psi,E_K(\lambda)\Psi}\nonumber\\
    &=\int_{[-\Lambda_M,\Lambda_M]}4\sin^2\left(\frac{\lambda t}{2}\right)\,\braket{\Psi,\mathrm{d}E_K(\lambda)\Psi}\nonumber\\
    &=t^2\int_{[-\Lambda_M,\Lambda_M]}\lambda^2\left[\frac{\sin\left(\frac{\lambda t}{2}\right)}{\frac{\lambda t}{2}}\right]^2\,\braket{\Psi,\mathrm{d}E_K(\lambda)\Psi}.
\end{align}
On the other hand, a direct graphical study of the function $\mathbb{R}\ni x\mapsto\sin^2x/x^2$ shows that the equation
\begin{equation}
   \left[\frac{\sin x}{x}\right]^2=\frac{1}{2}
\end{equation}
admits a unique couple of solutions $\pm x_0$, and $(\sin x/x)^2\geq1/2$ for all $x\in[-x_0,x_0]$. Therefore, defining $T_M:=2x_0/\Lambda_M$, we have
\begin{equation}
    \frac{\lambda t}{2}\leq x_0\qquad\text{for all }|t|\leq T_M,\;|\lambda|\leq \Lambda_M,
\end{equation}
and therefore
\begin{equation}
    \left[\frac{\sin\left(\frac{\lambda t}{2}\right)}{\frac{\lambda t}{2}}\right]^2\geq\frac{1}{2}\qquad\text{for all }|t|\leq T_M,\;|\lambda|\leq \Lambda_M,
\end{equation}
thus finally implying that, for all $|t|\leq T_M$,
\begin{align}
    \left\|(V(t)-I)\Psi\right\|^2&\geq \frac{t^2}{2}\int_{[-\Lambda_M,\Lambda_M]}\lambda^2\,\braket{\Psi,\mathrm{d}E_K(\lambda)\Psi}\nonumber\\
    &\geq M^2t^2\|\Psi\|^2.
\end{align}
We have proven that, for all $M>0$, there exists $T_M$ such that, for all $|t|\leq T_M$, $\|V(t)\Psi-\Psi\|\geq M|t|\|\Psi\|$; as such, (2) is false. This concludes the proof.
\end{proof}

\begin{proof}[Proof of Theorem~\ref{thm:liouville}]
We shall begin by proving the implication (i) above. Let $A\in\mathcal{L}(\hilb)$ satisfying the following conditions:
\begin{itemize}
    \item $A\Dom H\subset\Dom H$;
    \item the operator $[H,A]:\Dom H\subset\hilb\rightarrow\hilb$ is bounded, and its unique extension $\overline{[H,A]}$ to $\hilb$ is Hilbert--Schmidt.
\end{itemize}
We must prove that $A\in\Dom\mathbf{H}$. By Lemma~\ref{lemma:courbage}, this is equivalent to the existence of a constant $C_A>0$ such that, for all $t\in\mathbb{R}$,
\begin{equation}
    \|\mathbf{U}(t)A-A\|_{\rm HS}\leq C_A|t|.
\end{equation}
Let $(e_n)_{n\in\mathbb{N}}\subset\Dom H$ be a complete orthonormal set of $\hilb$. Such a basis always exists. Indeed, being the subset of a separable Hilbert space, $\Dom H$ is itself separable, which, by definition, means that it admits a dense countable subset, say $(\epsilon_n)_{n\in\mathbb{N}}$, which we can then orthonormalize via the Gram--Schmidt algorithm, thus obtaining our desired $(e_n)_{n\in\mathbb{N}}$. This is an orthonormal set, and it is complete since its span is dense in $\Dom H$ and thus in $\hilb$ as well.

Now, given any $(e_n)_{n\in\mathbb{N}}\in\Dom H$ and $s\in\mathbb{R}$, we have
\begin{align}
    \ii\frac{\mathrm{d}}{\mathrm{d}s}\mathbf{U}(s)Ae_n&=\ii\frac{\mathrm{d}}{\mathrm{d}s}\e^{-\ii sH}A\e^{\ii sH}e_n\nonumber\\
    &=H\e^{-\ii sH}A\e^{\ii sH}e_n-\e^{-\ii sH}AH\e^{\ii sH}e_n\nonumber\\
    &=\e^{-\ii sH}[H,A]\e^{\ii sH}e_n
\end{align}
where we have used the fact that both $\e^{-\ii sH}$ and $A$ map $\Dom H$ into itself. Since both sides of the equation above are continuous, we can now integrate it between $0$ and some $t>0$:
\begin{equation}
    \left(\mathbf{U}(t)A-A\right)e_n=-\ii\int_0^t\e^{-\ii sH}[H,A]\e^{\ii sH}e_n\;\mathrm{d}s,
\end{equation}
whence
\begin{equation}
   \left\| \left(\mathbf{U}(t)A-A\right)e_n\right\|\leq\int_0^t\|[H,A]\e^{\ii sH}e_n\|\;\mathrm{d}s,
\end{equation}
where the triangular inequality was applied to the integral. But then
\begin{align}
   \left\| \left(\mathbf{U}(t)A-A\right)e_n\right\|^2&\leq\left(\int_0^t\|[H,A]\e^{\ii sH}e_n\|\;\mathrm{d}s\right)^2\nonumber\\&=\left(\int_0^t1\times\|[H,A]\e^{\ii sH}e_n\|\;\mathrm{d}s\right)^2\nonumber\\
   &\leq\int_0^t 1\;\mathrm{d}\tau\int_0^t\|[H,A]\e^{\ii sH}e_n\|^2\;\mathrm{d}s\nonumber\\&=t\int_0^t\|[H,A]\e^{\ii sH}e_n\|^2\;\mathrm{d}s,
\end{align}
where we have applied the Cauchy--Schwarz inequality.

Now, for every $s\in\mathbb{R}$, $(\e^{-\ii sH}e_n)_{n\in\mathbb{N}}$ is again a complete orthonormal set in $\Dom H$, and we have
\begin{equation}
    \sum_{n\in\mathbb{N}}\|[H,A]\e^{\ii sH}e_n\|^2=\|\overline{[H,A]}\|_{\rm HS}^2,
\end{equation}
where we have used the fact that (a) $[H,A]$, being bounded, has a unique bounded extension to the whole $\hilb$; (b) such an extension is Hilbert--Schmidt. But then
\begin{align}
    \left\|\mathbf{U}(t)A-A\right\|^2_{\rm HS}&=\sum_{n\in\mathbb{N}} \left\| \left(\mathbf{U}(t)A-A\right)e_n\right\|^2\nonumber\\&\leq t\int_0^t\|\overline{[H,A]}\|_{\rm HS}^2\,\mathrm{d}s\nonumber\\&=t^2\|\overline{[H,A]}\|_{\rm HS}^2,
\end{align}
whence the claimed inequality holds with $C_A=\|\overline{[H,A]}\|_{\rm HS}$. This concludes the proof of the implication (i).

Let us now show the converse implication (ii). Let $A\in\Dom\mathbf{H}$, that is: there is $A'\in\mathcal{L}(\hilb)$ such that
\begin{equation}
    \lim_{t\to0}\left\|\frac{\mathbf{U}(t)A-A}{t}-A'\right\|_{\rm HS}=0,
\end{equation}
and, by definition (cf.~Eq.~\eqref{eq:hs_limit}), $A'=-\ii\mathbf{H}A$.

We want to show that this implies
\begin{itemize}
    \item $A\Dom H\subset\Dom H$;
    \item $[H,A]$ is bounded and its closure $\overline{[H,A]}$ is Hilbert--Schmidt.
\end{itemize}
But, since the operator norm $\|\cdot\|$ on $\mathcal{B}(\hilb)$ satisfies $\|\cdot\|\leq\|\cdot\|_{\rm HS}$ (cf.~Proposition~\ref{prop:bases}), for every $\psi\in\hilb$ we have
\begin{equation}
    \left\|\left(\frac{\mathbf{U}(t)A-A}{t}-A'\right)\psi\right\|\leq\left\|\frac{\mathbf{U}(t)A-A}{t}-A'\right\|\|\psi\|\leq\left\|\frac{\mathbf{U}(t)A-A}{t}-A'\right\|_{\rm HS}\|\psi\|,
\end{equation}
that is,
\begin{equation}
   \lim_{t\to0}\frac{\e^{-\ii tH}A\e^{\ii tH}-A}{t}\psi= \lim_{t\to0}\frac{\mathbf{U}(t)A-A}{t}\psi=A'\psi.
\end{equation}
If, in addition, $\psi\in\Dom H$, then by the defining properties of $\e^{\ii tH}$ and the boundedness of $A$ we have
\begin{equation}
    \lim_{t\to0}\e^{-\ii tH}A\frac{\e^{+\ii tH}-1}{t}\psi=\ii AH\psi.
\end{equation}
But then
\begin{align}
    \frac{\e^{-\ii tH}-1}{t}A\psi=\frac{\e^{-\ii tH}A\e^{\ii tH}-A}{t}\psi-\e^{-\ii tH}A\frac{\e^{\ii tH}-1}{t}\psi,
\end{align}
whence
\begin{equation}\label{eq:limit}
    \lim_{t\to0}\frac{\e^{-\ii tH}-1}{t}A\psi=A'\psi-\ii AH\psi.
\end{equation}
We have shown that, for every $\psi\in\Dom H$, $A\psi\in\Dom H$; that is, $A\Dom H\subset\Dom H$. Besides, by construction, the limit in the left-hand side of Eq.~\eqref{eq:limit} is equal to $-\ii HA\psi$, whence we also proved that, for all $\psi\in\Dom H$,
\begin{equation}
    \mathbf{H}A\psi=\ii A'\psi=[H,A]\psi,
\end{equation}
that is, the operator $[H,A]:\Dom H\subset\hilb\rightarrow\hilb$ agrees with the bounded, Hilbert--Schmidt operator $\mathbf{H}A$ on the dense subspace $\Dom H$. This means that $[H,A]$ is bounded, and its unique extension on the whole $\hilb$, $\overline{[H,A]}$, coincides with the Hilbert--Schmidt operator $\mathbf{H}A$. This concludes the proof.    
\end{proof}

While the characterization provided by Theorem~\ref{thm:liouville} is explicit and relatively simple, in practical applications one might not be able to compute the operator $\overline{[H,A]}$. However, we can recast the condition $\overline{[H,A]}\in\mathcal{L}(\hilb)$ in an equivalent, but conceptually simpler form, which only involves $[H,A]$.
\begin{theorem}\label{thm:equiv_domain}
The following characterization of $\Dom\mathbf{H}$, equivalent to the one in Theorem~\ref{thm:liouville}, holds: $A\in\Dom\mathbf{H}$ if and only if
\begin{itemize}
    \item [(i)] $A\Dom H\subset\Dom H$;
    \item [(ii)] There exists an orthonormal basis $(e_n)_{n\in\mathbb{N}}\subset\Dom H$ of $\hilb$ such that
    \begin{equation}\label{eq:hs_on_nice_domain}
        \sum_{n\in\mathbb{N}}\bigl\|[H,A]e_n\bigr\|^2<\infty.
    \end{equation}
\end{itemize}
\end{theorem}
\begin{proof}
As discussed in the proof of Theorem~\ref{thm:liouville}, there always exists one orthonormal basis $(e_n)_{n\in\mathbb{N}}\subset\Dom H$. By Theorem~\ref{thm:liouville}, to prove the desired property, we need to prove that, given $A\in\mathcal{L}(\hilb)$ such that $A\Dom H\subset\Dom H$, the condition $\overline{[H,A]}\in\mathcal{L}(\hilb)$ is equivalent to the inequality~\eqref{eq:hs_on_nice_domain} for one (and thus for all) orthonormal basis $(e_n)_{n\in\mathbb{N}}\subset\Dom H$. If $\overline{[H,A]}\in\mathcal{L}(\hilb)$, then Eq.~\eqref{eq:hs_on_nice_domain} simply follows from the fact that $\overline{[H,A]}e_n=[H,A]e_n$ for all $n\in\mathbb{N}$. Vice versa, assume that Eq.~\eqref{eq:hs_on_nice_domain} holds for some $(e_n)_{n\in\mathbb{N}}\subset\Dom H$. This means that the operator $[H,A]:\Dom H\subset\hilb\rightarrow\hilb$ is bounded, since, given $\psi\in\Dom H$,
    \begin{align}
        \left\|[H,A]\psi\right\|^2&=\sum_{n\in\mathbb{N}}\left|\braket{e_n,[H,A]\psi}\right|^2\nonumber\\
        &=\sum_{n\in\mathbb{N}}\left|\braket{[H,A]e_n,\psi}\right|^2\nonumber\\
        &\leq\|\psi\|^2\sum_{n\in\mathbb{N}}\left\|[H,A]e_n\right\|^2.
    \end{align}
    As such, $\overline{[H,A]}\in\mathcal{B}(\hilb)$ (by the BLT theorem, see Section~\ref{sec:prelimin}) and, again because $\overline{[H,A]}e_n=[H,A]e_n$, it is a Hilbert--Schmidt operator.
\end{proof}

\begin{remark}[Solvability of the Liouville--von Neumann equation]\label{rem:solvability_liouville}
Theorem~\ref{thm:liouville} and Theorem~\ref{thm:equiv_domain} completely characterize the Liouville superoperator $\mathbf{H}$ in terms of the corresponding Hamiltonian $H$, and finally allows us to ``almost'' identify Eq.~\eqref{eq:superschrodinger} with the Liouville--von Neumann equation~\eqref{eq:liouville}; namely, given a state $\rho_0$ satisfying the conditions in Theorem~\ref{thm:equiv_domain}, the equation
\begin{equation}\label{eq:liouville_corrected}
    \ii\frac{\mathrm{d}}{\mathrm{d}t}\rho(t)=\overline{[H,\rho(t)]}
\end{equation}
admits $\rho(t)=U(t)\rho_0U(t)^*$ as its unique solution. This is essentially the Liouville--von Neumann equation~\eqref{eq:liouville}, where the closure on the right-hand side ``takes care'' of domain issues. Aside from this technical point, Theorem~\ref{thm:equiv_domain} then provides a very simple recipe for the Liouville--von Neumann equation to admit a unique solution: this happens if and only if the initial state $\rho(0)=\rho_0$ enjoys the two properties listed in the statement of Theorem~\ref{thm:equiv_domain}\@. As we will discuss later (cf.\ Remark~\ref{rem:solvability_liouville2}), one can indeed find even easier, but only sufficient, conditions.
\end{remark}

\begin{remark}[Liouvillian in the vectorized representation]\label{rem:vectorization3}
    As pointed out in Remark~\ref{rem:vectorization2}, one could equivalently define the Liouville superoperator $\mathbf{H}$ in the vectorized representation $\mathcal{L}(\hilb)\simeq\hilb\otimes\hilb^*$ introduced in Remark~\ref{rem:vectorization}\@. In this representation, it can be shown (cf. Refs.~\cite{antoniou1998general,antoniou1999spectrum,antoniou2004quantum,ordonez2001explicit}) that the Liouvillian coincides with the following operator:
    \begin{equation}\label{eq:liouvillian_vectorized}
        \mathbf{H}=\overline{H\otimes I-I\otimes H},
    \end{equation}
    that is, the closure of the operator $H\otimes I-I\otimes H$ as defined on $\Dom H\otimes\Dom H$ (the tensor product \textit{now} being an algebraic one). In order to characterize explicitly the domain of the latter operator, one could use similar arguments as the ones in the proof of Theorem~\ref{thm:liouville}, adapted to this representation.

    We stress that the identification of $\mathbf{H}$ with the operator in Eq.~\eqref{eq:liouvillian_vectorized} does \textit{not} allow us to circumvent the problem of determining whether a (generally) infinite-rank Hilbert--Schmidt operator belongs to $\Dom\mathbf{H}$. Indeed, such operators are mapped into \textit{infinite} sums of tensor products, and therefore outside the algebraic tensor product $\Dom H\otimes\Dom H$ (which only includes finite sums). The two representations carry exactly the same amount of information, and it is ultimately equivalent to work within one or another.
\end{remark}

\subsection{Higher powers of the Liouville superoperator}

In applications, one might also be interested in analyzing higher powers of $\mathbf{H}$. We shall therefore consider the operator $\mathbf{H}^2:=\mathbf{H}\mathbf{H}$, which is defined by\footnote{Note that $\mathbf{H}^2$, the square of the Liouville superoperator associated with $H$, does \textit{not} coincide with the Liouville superoperator associated with the square $H^2$ of $H$.
}
\begin{align}\label{eq:domain_square}
    \Dom\mathbf{H}^2&=\left\{A\in\Dom\mathbf{H}:\mathbf{H}A\in\Dom\mathbf{H}\right\};\\
    \mathbf{H}^2A&=\mathbf{H}(\mathbf{H}A).
\end{align}
As the square of a self-adjoint superoperator, $\mathbf{H}^2$ is itself self-adjoint. Besides, Theorem~\ref{thm:liouville} allows us to explicitly express it in terms of $H$:

\begin{proposition}\label{prop:liouville^2}
Let $\mathbf{H}$ be the Liouville superoperator as per Definition~\ref{def:liouville}, and $H$ the corresponding Hamiltonian. Then
\begin{align}
    \Dom\mathbf{H}^2&=\Big\{A\in\mathcal{L}(\hilb):A\Dom H\subset\Dom H,\overline{[H,A]}\in\mathcal{L}(\hilb),\nonumber\\&\qquad\quad\;[H,A]\Dom H\subset\Dom H,\overline{[H,\overline{[H,A]}]}\in\mathcal{L}(\hilb)\Big\};\\
    \mathbf{H}^2A&=\overline{[H,\overline{[H,A]}]}.
\end{align}
\end{proposition}
\begin{proof}
    The claim follows directly by plugging in the explicit expressions for the domain and expression of $\mathbf{H}$ given by Theorem~\ref{thm:liouville} into Eq.~\eqref{eq:domain_square}, and by noticing that the condition $  \mathbf{H}A\,\Dom H\subset\Dom H$, i.e.\ $\overline{[H,A]}\Dom H\subset\Dom H$, simplifies to $[H,A]\Dom H\subset\Dom H$ since $\Dom [H,A]=\Dom H$.
\end{proof}
Following the footsteps of Theorem~\ref{thm:equiv_domain}, we would like to characterize $\Dom\mathbf{H}^2$ in an equivalent, but computationally simpler, way. This can be indeed obtained:
\begin{proposition}\label{prop:equiv_domain^2}
    The following characterization of $\Dom\mathbf{H}^2$, equivalent to the one in Proposition~\ref{prop:liouville^2}, holds: $A\in\Dom\mathbf{H}^2$ if and only if
    \begin{itemize}
        \item[(i)] $A\Dom H\subset\Dom H$ and $[H,A]\Dom H\subset\Dom H$;
        \item[(ii)] There exists an orthonormal basis $(e_n)_{n\in\mathbb{N}}\subset\Dom H^2$ of $\hilb$ such that
        \begin{equation}
            \sum_{n\in\mathbb{N}}\left\|[H,A]e_n\right\|^2<\infty\quad\text{and}\qquad \sum_{n\in\mathbb{N}}\left\|[H,[H,A]]e_n\right\|^2<\infty.
        \end{equation}
    \end{itemize}
\end{proposition}
\begin{proof}[Proof]
    Following the same argument as at the beginning of the proof of Theorem~\ref{thm:equiv_domain}, we can always find an orthonormal basis $(e_n)_{n\in\mathbb{N}}\subset\Dom H^2$ of $\hilb$. Now, the assumption $A\Dom H\subset \Dom H$ again implies that $[H,A]$ is defined on $\Dom H$. In turn, since $e_n\in\Dom H^2$ and $\Dom H^2\subset\Dom H$, we have
    \begin{equation}
        \overline{[H,A]}e_n=[H,A]e_n.
    \end{equation}
    Besides, the additional assumption $[H,A]\Dom H\subset\Dom H$ ensures the two following properties:
    \begin{itemize}
    \item[(a)] the operator $[H,[H,A]]=H[H,A]-[H,A]H$ is well-defined on $\Dom H^2$;
        \item[(b)]the operator $[H,\overline{[H,A]}]=H\overline{[H,A]}-\overline{[H,A]}H$ is well-defined on $\Dom H$, and its action on $\Dom H^2$ coincides with the one of $[H,[H,A]]$.         
    \end{itemize}
    To prove (a), notice that, since $[H,A]\Dom H\subset\Dom H$, $H[H,A]$ is defined on $\Dom H$ and thus, a fortiori, on $\Dom H^2$ (since again $\Dom H^2\subset\Dom H$); besides, $[H,A]H$ is defined on $\Dom H^2$ since $H\Dom H^2\subset\Dom H$ and $[H,A]$ is defined on $\Dom H$. To prove (b), it suffices to show that $\overline{[H,A]}\Dom H=[H,A]\Dom H$, thus $H\overline{[H,A]}$ is well-defined on $\Dom H$ and coincides with $H[H,A]$ on it; however, $\overline{[H,A]}H$ is well-defined on the whole $\Dom H$ (since $\overline{[H,A]}$ is bounded) and, again because of $H\Dom H^2\subset\Dom H$, its action on $\Dom H^2$ coincides with the one of $[H,A]H$.   
    
    Therefore, since $e_n\in\Dom H^2$,
    \begin{equation}
        \overline{[H,\overline{[H,A]}]}e_n=\overline{[H,[H,A]]}e_n=[H,[H,A]]e_n.
    \end{equation}
    But then, following the very same argument as in the proof of Theorem~\ref{thm:equiv_domain}, one readily sees that
    \begin{itemize}
        \item assuming $A\Dom H\subset\Dom H$, then
        \begin{equation}
            \overline{[H,A]}\in\mathcal{L}(\hilb)\;\text{ if and only if }\;\sum_{n\in\mathbb{N}}\|[H,A]e_n\|^2<\infty;
        \end{equation}
           \item assuming $[H,A]\Dom H\subset\Dom H$, then 
           \begin{equation}
               \overline{[H,\overline{[H,A]}]}\in\mathcal{L}(\hilb)\;\text{ if and only if }\;\sum_{n\in\mathbb{N}}\|[H,[H,A]]e_n\|^2<\infty,
           \end{equation}
    \end{itemize}
    whence the claimed property.
\end{proof}
Similar characterizations may be obtained, in principle, for higher powers $\mathbf{H}^k$ of the Liouville superoperator.

\section{Cores of the Liouville superoperator}\label{sec:cores}

Next we focus on some domains of \textit{essential} self-adjointness (cores) for the Liouville superoperator $\mathbf{H}$, that is, proper subspaces $\mathcal{D}\subset\Dom\mathbf{H}$ such that the restriction of $\mathbf{H}$ to $\mathcal{D}$ satisfies $\overline{\mathbf{H}\restriction\mathcal{D}}=\mathbf{H}$. Typically, these domains entail stronger conditions on Hilbert--Schmidt operators $A\in\mathcal{L}(\hilb)$, to the benefit of being easier to verify in practice. Some examples, often found in the literature (see the subsequent discussion in Remark~\ref{rem:biblio_cores}), essentially involve replacing the condition $\overline{HA-AH}\in\mathcal{L}(\hilb)$ with two separate conditions on the operators $HA,AH$.

We will need the following, important sufficient criterion for essential self-adjointness:
\begin{proposition}[Nelson]\label{prop:nelson}
    Let $\mathcal{K}$ be a Hilbert space and $\{V(t)\}_{t\in\mathbb{R}}$ be a strongly continuous one-parameter unitary group with infinitesimal generator $K$. Let the subspace $\mathcal{D}\subset\hilb$ satisfy the following properties:
    \begin{itemize}
        \item[(i)] $\mathcal{D}\subset\Dom K$;
        \item[(ii)] $\mathcal{D}$ is dense in $\hilb$;
        \item[(iii] $V(t)\mathcal{D}\subset\mathcal{D}$ for all $t$.
    \end{itemize}
    Then $\mathcal{D}$ is a core for $K$.
\begin{proof}
    This is proven for instance in Ref.~\cite[Proposition~1.7]{Engel-Nagel2000}\@.
\end{proof}
\end{proposition}
In other words, any dense subspace of the domain of the infinitesimal generator that is invariant under the action of the corresponding unitary group is a core.

The following lemma will also be useful:
\begin{lemma}\label{lemma:1}
    Let $H$ self-adjoint and $A\in\mathcal{L}(\hilb)$. Then the following conditions are equivalent:
    \begin{itemize}
        \item[(i)] $\overline{HA}\in\mathcal{L}(\hilb)$;
        \item[(ii)] $HA\in\mathcal{L}(\hilb)$.
    \end{itemize}
\end{lemma}
\begin{proof}
(ii)$\!\implies\!$(i) directly follows from the fact that, if $HA\in\mathcal{L}(\hilb)$, then $HA\in\mathcal{B}(\hilb)$; but then, as discussed in Section~\ref{sec:prelimin}, $HA$ is closed, thus $\overline{HA}=HA$.

Let us prove (i)$\!\implies\!$(ii). Let $\overline{HA}\in\mathcal{L}(\hilb)$, which in particular implies $\overline{HA}\in\mathcal{B}(\hilb)$. This implies that $HA$, with domain
\begin{equation}
    \Dom HA=\left\{\varphi\in\hilb:A\varphi\in\Dom H\right\},
\end{equation}
is bounded, since an unbounded operator cannot have a bounded closure, and densely defined, since the closure of a bounded operator is its unique extension to the closure of its domain (cf.\ Section~\ref{sec:prelimin}). 
We have to prove that, actually, $\Dom HA=\mathcal{H}$. To this purpose, let $\psi\in\mathcal{H}$. Since $HA$ is densely defined, there exists $(\psi_n)_n\subset\hilb$ such that
   \begin{itemize}
       \item $\psi_n\to\psi$;
       \item $A\psi_n\in\Dom H$.
   \end{itemize}
   Since $A$ is bounded, we also have $A\psi_n\to A\psi$. Besides, since $HA$ is also bounded, the sequence $(HA\psi_n)_n$ is Cauchy and therefore convergent, that is, defining $\xi_n:=A\psi_n\in\Dom H$ and $\xi:=A\psi\in\mathcal{H}$, we have the following:
       \begin{itemize}
           \item $\xi_n\to\xi$;
           \item $H\xi_n$ converges.
       \end{itemize}
    But, since $H$ is closed, necessarily $\xi\in\Dom H$ and $H\xi_n\to H\xi$. This, in particular, proves $\xi=A\psi\in\Dom H$, that is, $\psi\in\Dom HA$. Since $\psi\in\mathcal{H}$ was arbitrary, we have proven $\Dom HA=\hilb$. This shows that $HA\in\mathcal{B}(\hilb)$ and thus again $HA=\overline{HA}$.

    We have proven that $HA\in\mathcal{B}(\hilb)$ if and only if $\overline{HA}\in\mathcal{B}(\hilb)$; that is, the two operators coincide whenever either of them is bounded. A fortiori, this immediately implies that whenever one of them is Hilbert--Schmidt, the other one is.
\end{proof}

In addition, we will need the following auxiliary lemma:
\begin{lemma}\label{lemma:2}
    Let $H$ self-adjoint, $A\in\mathcal{B}(\hilb)$, $HA$ densely defined. Then $(AH)^*=HA^*$.
\end{lemma}
\begin{proof}
    It is known (see e.g.~\cite[Problem~2.3]{teschl2014mathematical}) that, if $H$ is densely defined, $A\in\mathcal{B}(\hilb)$, and $HA$ densely defined, we have $(AH)^*=H^*A^*$. In our case $H=H^*$ and thus $(AH)^*=HA^*$.
\end{proof}

We are now ready to present a core $\mathcal{D}$ of the Liouvillian:
\begin{theorem}\label{thm:core1}
     The subspace $\mathcal{D}$ is defined by any of the following equivalent expressions:
    \begin{align}\label{eq:core1}
        \mathcal{D}&=\left\{A\in\mathcal{L}(\hilb):HA,HA^*\in\mathcal{L}(\hilb)\right\}\\\label{eq:core2}
        &=\left\{A\in\mathcal{L}(\hilb):HA,\overline{AH}\in\mathcal{L}(\hilb)\right\}\\\label{eq:core3}
        &=\left\{A\in\mathcal{L}(\hilb):\overline{HA},\overline{AH}\in\mathcal{L}(\hilb)\right\}
    \end{align}   
    is a core for $\mathbf{H}$. Besides, for every $A\in\mathcal{D}$,
    \begin{align}\label{eq:liouville_on_core1}
        \mathbf{H}A&=HA-(HA^*)^*\\\label{eq:liouville_on_core2}
        &=HA-\overline{AH}\\\label{eq:liouville_on_core3}
        &=\overline{HA}-\overline{AH}.
    \end{align}
\end{theorem}
\begin{proof}
    To begin with, let us show that the three expressions for the core are actually equivalent. Assume $A,HA\in\mathcal{L}(\hilb)$, thus in particular $A,HA\in\mathcal{B}(\hilb)$. By Lemma~\ref{lemma:2}, the densely defined operator $AH$ satisfies $(AH)^*=HA^*$, so we must equivalently verify that, under our assumptions, $(AH)^*\in\mathcal{L}(\hilb)$ if and only if $\overline{AH}=(AH)^{**}\in\mathcal{L}(\hilb)$. Assume $(AH)^*\in\mathcal{L}(\hilb)$; then its adjoint $(AH)^{**}=\overline{AH}$ is in $\mathcal{L}(\hilb)$ since the adjoint of a Hilbert--Schmidt operator is still Hilbert--Schmidt. Vice versa, let $\overline{AH}\in\mathcal{L}(\hilb)$. Then its adjoint $\overline{AH}^*=\overline{(AH)^*}=\overline{HA^*}$ is also in $\mathcal{L}(\hilb)$. But then, applying Lemma~\ref{lemma:1}, $HA^*=(AH)^*\in\mathcal{L}(\hilb)$ as well. This proves the equivalence between the expressions~\eqref{eq:core1} and~\eqref{eq:core2}\@. The equivalence between the expressions~\eqref{eq:core2} and~\eqref{eq:core3} follows directly by Lemma~\ref{lemma:1}\@.

    To prove that $\mathcal{D}$ is a core for $\mathbf{H}$, we shall use Nelson's criterion (Proposition~\ref{prop:nelson}), which is automatically valid for superoperators on $\mathcal{L}(\hilb)$. That is, we must check the following three conditions:
\begin{itemize}
        \item[(i)] $\mathcal{D}\subset\Dom\mathbf{H}$;
        \item[(ii)] $\mathcal{D}$ is dense in $\mathcal{L}(\hilb)$;
        \item[(iii)] $\mathbf{U}(t)\mathcal{D}\subset\mathcal{D}$ for all $t$.\end{itemize}

        To prove (i), let $A\in\mathcal{D}$, that is, $HA$, $HA^*\in\mathcal{L}(\hilb)$. The first assumption, in particular, requires $HA\in\mathcal{B}(\hilb)$, thus $A\hilb\subset\Dom H$ and thus, a fortiori, $A\Dom H\subset\Dom H$. Besides, again because $HA\in\mathcal{B}(\hilb)$, the operator $[H,A]=HA-AH$, which is correctly defined on $\Dom H$, satisfies
        \begin{equation}
            (HA-AH)^*=(HA)^*-(AH)^*
        \end{equation}
        and thus
        \begin{equation}
            \overline{HA-AH}=(HA-AH)^{**}=HA-\overline{AH}.
        \end{equation}
        To prove (ii), recall that finite-rank operators are dense in $\mathcal{L}(\hilb)$ (Proposition~\ref{prop:finite-rank}); as such, it suffices to show that any finite-rank operator in $\mathcal{L}(\hilb)$ can be approximated arbitrarily well by a finite-rank operator in $\mathcal{D}$. Let $A\in\mathcal{L}(\hilb)$ such that $\dim\Ran A=N\in\mathbb{N}$. Thus, we can find two orthonormal sets $(u_j)_{j=1,\dots,N},(v_j)_{j=1,\dots,N}\subset\hilb$ and a complex $N$-tuple $(\alpha_j)_{j=1,\dots,N}\subset\mathbb{C}$ such that
        \begin{equation}
            A=\sum_{j=1}^N\alpha_j\braket{v_j,\cdot}u_j,\quad\text{and therefore}\quad A^*=\sum_{j=1}^N\overline{\alpha_j}\braket{u_j,\cdot}v_j
        \end{equation}
        Since $\Dom H$ is dense in $\hilb$, we can always find two families of normalized vectors $(\tilde{u}_j)_{j=1,\dots,N}\subset\Dom H$, $(\tilde{v}_j)_{j=1,\dots,N}\subset\Dom H$, whose elements are arbitrarily close to the elements of $(u_j)_{j=1,\dots,N}$ and $(v_j)_{j=1,\dots,N}$, respectively. In particular, given $\epsilon>0$, we can choose
        \begin{equation}
            \|u_j-\tilde{u}_j\|\leq\frac{\epsilon}{2N|\alpha_j|},\qquad\|v_j-\tilde{v}_j\|\leq\frac{\epsilon}{2N|\alpha_j|}\label{eq:dense_vectors}
        \end{equation}
        for all $j=1,\dots,N$. Defining
        \begin{equation}
            \tilde{A}=\sum_{j=1}^N\alpha_j\braket{\tilde{v}_j,\cdot}\tilde{u}_j,
        \end{equation}
        we then have
        \begin{align}
            \|A-\tilde{A}\|_{\rm HS}&=\bigg\|\sum_{j=1}^N\alpha_j\Bigl(\braket{v_j-\tilde{v}_j,\cdot}u_j+\braket{\tilde{v}_j,\cdot}(u_j-\tilde{u}_j)\Bigr)\bigg\|_{\rm HS}\nonumber\\
            &\leq\sum_{j=1}^N|\alpha_j|\Bigl(\left\|\braket{v_j-\tilde{v}_j,\cdot}u_j\right\|_{\rm HS}+\left\|\braket{\tilde{v}_j,\cdot}(u_j-\tilde{u}_j)\right\|_{\rm HS}\Bigr)\nonumber\\
            &=\sum_{j=1}^N|\alpha_j|\Bigl(\|v_j-\tilde{v}_j\|+\|u_j-\tilde{u}_j\|\Bigr)\nonumber\\
            &\leq\sum_{j=1}^N\left(\frac{\epsilon}{2N}+\frac{\epsilon}{2N}\right)=\epsilon,
        \end{align}
        where the second step uses the triangle inequality, the third step uses that the vectors are normalized, and the last step follows from Eq.~\eqref{eq:dense_vectors}\@.
        Besides, since $\tilde{v}_j,\tilde{u}_j\in\Dom H$, we simply have
        \begin{equation}
            H\tilde{A}=\sum_{j=1}^N\alpha_j\braket{\tilde{v}_j,\cdot}H\tilde{u}_j,\qquad H\tilde{A}^*=\sum_{j=1}^N\overline{\alpha_j}\braket{\tilde{u}_j,\cdot}H\tilde{v}_j,
        \end{equation}
        and clearly $H\tilde{A},H\tilde{A}^*\in\mathcal{L}(\hilb)$. This proves (ii).

        Finally, let us prove (iii). Using the fact that $HA,HA^*\in\mathcal{L}(\hilb)$ and that $U(t)$ commutes with its generator $H$, we have
        \begin{align}
            H\mathbf{U}(t)A&=HU(t)AU(t)^*\nonumber\\&=U(t)HAU(t)^*=\mathbf{U}(t)(HA)
         \end{align}
        and similarly for $H\mathbf{U}(t)A^*$, thus showing that $\mathbf{U}(t)A\in\mathcal{D}$ as well. This completes the proof.
\end{proof}

\begin{remark}[Solvability of the Liouville--von Neumann equation II]\label{rem:solvability_liouville2}
    The expression~\eqref{eq:core1} is particularly convenient when one is interested in self-adjoint elements $A=A^*$ of the domain---given $A=A^*\in\mathcal{L}(\hilb)$, clearly $A\in\mathcal{D}$ if and only if the sole condition $HA\in\mathcal{L}(\hilb)$ is satisfied. In particular, this is the only condition one has to verify for density operators $\rho$ to be in $\mathcal{D}$: a sufficient condition for the Liouville--von Neumann equation~\eqref{eq:liouville_corrected} admitting a unique solution is that the initial condition $\rho(0)=\rho_0$ satisfies $H\rho_0\in\mathcal{L}(\hilb).$
\end{remark}

\begin{remark}[$\mathcal{D}$ is a proper subspace of $\Dom\mathbf{H}$]\label{rem:proper}
We remark that $\mathcal{D}$ is, in general, a \textit{proper} subspace of $\Dom\mathbf{H}$. As a simple example, suppose that the operator $H$ is positive and has a Hilbert--Schmidt resolvent, whence $H^{-1}\in\mathcal{L}(\hilb)$. Then
\begin{itemize}
    \item $H^{-1}\in\Dom\mathbf{H}$ since $H^{-1}\Dom H\subset\Dom H$ (in fact, $H^{-1}\hilb\subset\Dom H$), and $[H,H^{-1}]$ is just the restriction of the null operator to $\Dom H$; therefore, $\overline{[H,H^{-1}]}$ is the null operator, which is in $\mathcal{L}(\hilb)$;
    \item however, $H^{-1}\notin\mathcal{D}$. Indeed, $HH^{-1}=I$ and $H^{-1}H=I\restriction\Dom H$, and the identity operator is never compact (nor, a fortiori, Hilbert--Schmidt) in any infinite-dimensional space.
\end{itemize}
Therefore, the restriction of $\mathbf{H}$ to $\mathcal{D}$ is essentially self-adjoint but not self-adjoint.
\end{remark}

Finally, we introduce a smaller---but computationally simpler---core $\mathcal{D}_0$ of the Liouville superoperator:
\begin{proposition}\label{prop:core4}
    The subspace $\mathcal{D}_0$ defined as follows:
    \begin{equation}\label{eq:core4}
        \mathcal{D}_0=\biggl\{A=\sum_{j=1}^N\alpha_j\braket{\psi_j,\cdot}\varphi_j:\;N\in\mathbb{N},\,(\alpha_j)_{j=1,\dots,N}\subset\mathbb{C},(\psi_j)_{j=1,\dots,N},\,(\varphi_j)_{j=1,\dots,N}\subset\Dom H\biggr\}
    \end{equation}
    is a core for $\mathbf{H}$. Besides, for every $A\in\mathcal{D}_0$,
    \begin{equation}\label{eq:liouville_on_core4}
        \mathbf{H}A=\sum_{j=1}^N\alpha_j\Bigl(\braket{\psi_j,\cdot}H\varphi_j-\braket{H\psi_j,\cdot}\varphi_j\Bigr).
    \end{equation}
\end{proposition}
\begin{proof}
    We will again invoke Nelson's criterion (Proposition~\ref{prop:nelson}). To this end, notice that in the proof of Theorem~\ref{thm:core1} we already showed that $\mathcal{D}_0\subset\mathcal{D}\subset\Dom\mathbf{H}$, and that $\mathcal{D}_0$ is dense, which shows that the first two conditions of the criterion are satisfied. We only have to show $\mathbf{U}(t)\mathcal{D}_0\subset\mathcal{D}_0$. Given $A\in\mathcal{D}_0$ represented as in Eq.~\eqref{eq:core4}, we have 
    \begin{equation}
        \mathbf{U}(t)A=\sum_{j=1}^N\alpha_j\braket{U(t)\psi_j,\cdot}U(t)\varphi_j,
    \end{equation}
    whence the claim follows from the known property $U(t)\Dom H\subset\Dom H$. Finally, to compute $\mathbf{H}A$, we can just use Eq.~\eqref{eq:liouville_on_core1}\@. We have
    \begin{align}
        HA&=\sum_{j=1}^N\alpha_j\braket{\psi_j,\cdot}H\varphi_j;\\
        HA^*&=\sum_{j=1}^N\overline{\alpha_j}\braket{\varphi_j,\cdot}H\psi_j,
    \end{align}
    and, from the last equation,
    \begin{equation}
        \overline{AH}=(HA^*)^*=\sum_{j=1}^N\alpha_j\braket{H\psi_j,\cdot}\varphi_j,
    \end{equation}
    whence Eq.~\eqref{eq:liouville_on_core4} follows.
\end{proof}

\begin{remark}\label{rem:vectorization4}
    Proposition~\ref{prop:core4} is, of course, compatible with the fact that, in the vectorized representation, $\mathbf{H}=\overline{H\otimes I-I\otimes H}$ (cf.\ Remark~\ref{rem:vectorization3}). Indeed, in this representation, $\mathbf{H}$ is by construction essentially self-adjoint on the algebraic tensor product $\Dom H\otimes\Dom H$, which is nothing but the vectorized representation of the space $\mathcal{D}_0$.
\end{remark}

\begin{remark}[Cores of the Liouvillian in the literature]\label{rem:biblio_cores}
    Expressions like the ones discussed in Theorem~\ref{thm:core1}, both for the core as well as the action of $\mathbf{H}$ on it, can be encountered in the literature. In Ref.~\cite[Theorem~3.4]{prugovevcki1975semi}, it is shown (in the broader context of semigroups of linear transformations on Banach spaces) that $\mathcal{D}$, expressed as in Eq.~\eqref{eq:core3}, is contained in $\Dom\mathbf{H}$, and that $\mathbf{H}$ acts on it as in Eq.~\eqref{eq:liouville_on_core3}\@. The same core and expressions are later discussed in Ref.~\cite[Lemma~8.5]{prugovecki1982quantum}, where it is claimed that, as a consequence of Ref.~\cite[Theorem~3.4]{prugovevcki1975semi}, $\mathbf{H}$ is self-adjoint on $\mathcal{D}$ (however, this is in general incorrect since, as seen in Remark~\ref{rem:proper}, $\mathcal{D}$ is generally a \textit{proper} subspace of $\Dom\mathbf{H}$). The same author shows, in Ref.~\cite[Theorem~A1]{prugovevcki1972scattering}, that the domain $\mathcal{D}_0$ discussed in Proposition~\ref{prop:core4} is a subspace of $\Dom\mathbf{H}$.       
        
    In Ref.~\cite{spohn1975spectral}, it is claimed that, as a consequence of Ref.~\cite[Theorem~1]{spohn1975spectral}, $\mathcal{D}$ as expressed in Eq.~\eqref{eq:core2} (but with the closure in the condition $\overline{AH}\in\mathcal{L}(\hilb)$ being erroneously missing), is a core of $\mathbf{H}$; the author also hints at an alternative proof involving Nelson's theorem. Finally, in Ref.~\cite{hirokawa1993rigorous} the essential self-adjointness of the Liouvillian (but in a different Liouville space involving a temperature-dependent inner product) on a domain equivalent to Eq.~\eqref{eq:core2} is proven, by using an analogous strategy as in Theorem~\ref{thm:core1}\@.
\end{remark}
    
\section{Concluding remarks}\label{sec:conclusion}

In this paper, we established a complete characterisation of the Liouville--von Neumann equation for unbounded Hamiltonians. Since this is a very fundamental subject, we felt that it was important to summarise the existing literature (which is surprisingly scarce, scattered, and partly incomplete or inconsistent) in a self-contained way with textbook character. In this spirit, we started by presenting and discussing a complete characterization of the domain and the action of the Liouville superoperator (Theorem~\ref{thm:liouville}), first found in Ref.~\cite{courbage1982mathematical} but seemingly not well-known. Incidentally, this result also shows the surprisingly subtle role of operator closures. Starting from this theorem, we have developed some novel, simpler-to-check characterizations of the self-adjointness domain of the Liouville superoperator (Theorem~\ref{thm:equiv_domain}) and of its square (Propositions~\ref{prop:liouville^2} and~\ref{prop:equiv_domain^2}). Finally, we presented some essential self-adjointness domains of the Liouvillian (Theorem~\ref{thm:core1} and Proposition~\ref{prop:core4}) that are often found, but without explicit proofs, in the literature.

Aside from their fundamental and pedagogical interest, these results are of practical importance for quantum technologies. Unbounded quantum Hamiltonians are often naturally encountered in applications---for instance, whenever discussing the interaction of quantum mechanical matter with bosonic fields; such operators cannot provide energy scales valid for all states of the system. Consequently, energy scales are directly determined by---and dependent on---the very choice of the initial state. To this end, pure states $\psi$ usually need to be in the domain of the Hamiltonian $H$ (and potentially higher moments), and energy scales depend on quantities like $\|H^k\psi\|$. The Hilbert--Schmidt setting allows one to generalize such relations, in principle, to mixed states $\rho$, with energy scales now depending on quantities like $\|\mathbf{H}^k\rho\|_{\rm HS}$. In such a case, the mathematical results listed above provide simple and ready-to-use recipes to check whether the desired initial state of a quantum system satisfies or violates the required constraint. This is particularly useful for applications where \emph{infinite-rank} density operators (such as Gibbs states) need to be considered, ranging from perturbative analysis of quantum statistics to control theory of thermal systems~\cite{hahn2024efficiency}\@. The results presented in this paper are especially useful for such states.
\medskip

\textit{Acknowledgments.} D.L. acknowledges financial support by Friedrich-Alexander-Universit\"at Erlangen-N\"urnberg through the funding program ``Emerging Talent Initiative'' (ETI). A.H. was partially supported by the Sydney Quantum Academy.

\bibliographystyle{prsty-title-hyperref}
\bibliography{spinbosondd.bib}

\end{document}